\pdfoutput=1
\documentclass[11pt]{article}
\usepackage{subfig}
\usepackage{amssymb,amsmath,amsthm}
\usepackage{graphics,epsfig}
\usepackage{algorithmic}
\usepackage{algorithm}
\usepackage[margin=1 in]{geometry}

\newtheorem{lem}{Lemma}[section]

\newtheorem{coro}[lem]{Corollary}
\newtheorem{thm}[lem]{Theorem}

\title{An Efficient Approximation Algorithm for the Steiner Tree Problem}
\author{Chi-Yeh~Chen % <-this % stops a space
\\ Department of Computer Science and Information
Engineering, \\ National Cheng Kung University, \\
Taiwan, ROC. \\
chency@csie.ncku.edu.tw.}
\begin{document}

\maketitle
\begin{abstract}
The Steiner tree problem is one of the classic and most fundamental $\mathcal{NP}$-hard problems: given an arbitrary weighted graph, seek a minimum-cost tree spanning a given subset of the vertices (terminals). Byrka \emph{et al}. proposed a $1.3863+\epsilon$-approximation algorithm in which the linear program is solved at every iteration after contracting a component. Goemans \emph{et al}. shown that it is possible to achieve the same approximation guarantee while only solving hypergraphic LP relaxation once. However, optimizing hypergraphic LP relaxation exactly is strongly NP-hard. This article presents an efficient two-phase heuristic in greedy strategy that achieves an approximation ratio of $1.4295$. 
\end{abstract}

\section{Introduction}\label{sec:Introduction}
The \textit{Steiner tree} problem is one of the classic and most fundamental $\mathcal{NP}$-hard problems. Given an arbitrary weighted graph with a distinguished vertex subset, the Steiner tree problem asks for a shortest tree spanning the distinguished vertices. 
This problem is widely used in many fields, such as VLSI routing~\cite{Kahng95}, wireless communications~\cite{Min2006,Liang:2002}, transportation~\cite{Hwang92}, wirelength estimation~\cite{Caldwell98}, and network routing~\cite{Korte90}.
The Steiner tree problem is $\mathcal{NP}$-hard even in the very special cases of Euclidean or rectilinear metrics~\cite{garey02}.
In fact, it is $\mathcal{NP}$-hard to approximate the Steiner tree problem within a factor $96/95$~\cite{Chlebik08}.
Hence, an approximation algorithm with a small and provable guarantee is thirsted by researchers.
Recall that an $\alpha$-approximation algorithm for a minimization problem is a polynomial-time algorithm that finds approximate solutions to $\mathcal{NP}$-hard optimization problems with cost at most $\alpha$ times the optimum value.

Arora~\cite{Arora98} established that Euclidean and rectilinear minimum-cost Steiner trees can be approximated in polynomial time within a factor of $1 + \epsilon$ for any constant $\epsilon>0$. For arbitrary weighted graphs, a sequence of improved approximation algorithms appeared in the literatures~\cite{Takahashi80, Zel93, Berman94, Zel96, Pro97, Karpinski97, Hougardy99, Robins05, Byrka13} and the best approximation ratio achievable within polynomial time was improved from 2 to 1.39.

Byrka \emph{et al}. proposed an LP-based approximation algorithm that achieves approximation ratio of $\ln4+\epsilon$ for general graphs~\cite{Byrka13}. However, the linear program is solved at every iteration after contracting a component. Goemans \emph{et al}.~\cite{Goemans2012} shown that it is possible to achieve the same approximation guarantee while only solving hypergraphic LP relaxation once. However, optimizing hypergraphic LP relaxation exactly is strongly NP-hard~\cite{Goemans2012}. Borchers and Du~\cite{Borchers97} show that $\rho_k \leq 1+\left\lfloor \log_{2}k\right\rfloor^{-1}$ where $\rho_k$ is the worst-case ratio of the cost of optimal $k$-restricted Steiner tree to the cost of optimal Steiner tree. We may therefore choose $k=2^{1/\epsilon}$ appropriately to obtain a $1+\epsilon$ approximation to hypergraphic LP relaxation, for any $\epsilon>0$. The number of variables and constraints will consequently be more than $n^{2^{1/\epsilon}}$ where $n$ is the number of terminals~\cite{Feldmann2016}.
%%%%%%%%%%%%%%%%%%%%%%%%%%%%%%%%%%%%%%%%%%%%%%%%%%%%%%%%%%%%%%
% section Preliminaries
%%%%%%%%%%%%%%%%%%%%%%%%%%%%%%%%%%%%%%%%%%%%%%%%%%%%%%%%%%%%%%
\section{Notation and Preliminaries}\label{sec:Preliminaries}
Given a graph $G=(V,E)$ with nonnegative edge costs (or weights) $cost:E\rightarrow \mathbb{R}^{+}$ and a subset $R \subseteq V$ of terminals of the vertices of $G$, the Steiner tree problem asks for a minimum-cost Steiner tree spanning $R$. 
Any tree in $G$ spanning $R$ is called a Steiner tree, and any non-terminal vertices contained in a Steiner tree are referred to as Steiner points.
The cost of a tree is the sum of its edge costs.
The graph $G$ is assumed to be a complete graph and let $G_R$ be a complete graph that induced by $R$. 

For any graph $H$, we denote by $MST(H)$ a minimum spanning tree of a graph $H$ and by $cost(H)$ the sum of the costs of all edges in $H$.
We thus abbreviate $mst(H)=cost(MST(H))$, i.e., the cost of a minimum spanning tree of $H$.

A \textit{terminal-spanning tree} is a Steiner tree that does not contain any Steiner points.
Let $mst$ be the cost of minimum terminal-spanning tree $MST(G_R)$.
A minimum-cost Steiner tree spanning subset $R'\subset R$ in which all terminals are leaves is called a \textit{full component}.
Any Steiner tree can be decomposed into full components by splitting all the non-leaf terminals~\cite{Robins05}.
Our algorithm will starts with a minimum-cost terminal spanning tree, and iteratively adds full components to improve it.
Any full component is assumed to have its own copy of each Steiner point so that full components chosen by our algorithm do not share Steiner points.

Let $\Gamma(K)$ be the terminal set of a given full component $K$.
Let $E_{0}(R')$ be the set of zero-cost edges in which all edges connect all pairs of terminals in $R'$. 
For brevity, let $E_{0}(H)=E_{0}(\Gamma(H))$. We call a Steiner tree $S$ is a \textit{well solution} if $\left|\Gamma(K_{i})\cap \Gamma(K_{j})\right|\leq 1$ for any two full components $K_{i}$ and $K_{j}$ in $S$.
Let $Loss(K)$ be the minimum-cost sub-forest of $K$. A simple method of computing $Loss(K)$ is given by the following lemma.
\begin{lem}\label{lem:lem-2}
{\rm\cite{Robins05}}. For any full component $K$, $Loss(K)=MST(K\cup E_0(K))-E_0(K)$.
\end{lem}

We denote the cost of $Loss(K)$ by $loss(K)$. Let $\mathcal{C}[K]$ be a loss-contracted full component that can be obtained by collapsing each connected component of $Loss(K)$ into a single node.
We denote by $Opt_{k}$ an optimal $k$-restricted Steiner tree. Let $opt_{k}$ and $loss_{k}$ be the cost and loss of $Opt_{k}$, respectively. Let $opt$ be the cost of the optimal Steiner tree.
For brevity, this article uses $T/ E_{0}(R')$ to denote the minimum spanning tree of $T\cup E_{0}(R')$ for $R'\subset R$.

The \textit{gain} of a full component $K$ with respect to $T$ is defined as
\begin{eqnarray*}
gain_{T}(K) = cost(T)-mst(T\cup E_{0}(K))-cost(K),
\end{eqnarray*}
and the \textit{load} of of a full component $K$ with respect to $T$ is defined as
\begin{eqnarray*}
load_{T}(K) = cost(K)+mst(T\cup E_{0}(K))- cost(T).
\end{eqnarray*}
Let $\Psi_{T_{1},T_{2}}(K)=cost(T_{1})-cost(T_{2})-mst(T_{1}\cup E_{0}(K))+mst(T_{2}\cup E_{0}(K))$.
The following lemma shows that if no full component can improve a terminal-spanning tree $T$, then $cost(T)\leq opt_{k}$.
\begin{lem}\label{lem08}
{\rm\cite{Robins05}}. Let $T$ be a terminal-spanning tree; if $gain_{T}(K)\leq 0$ for any k-restricted full component $K$, then $cost(T)\leq opt_{k}$.
\end{lem}

%%%%%%%%%%%%%%%%%%%%%%%%%%%%%%%%%%%%%%%%%%%%%%%%%%%%%%%%%%%%%%
% section The algorithm
%%%%%%%%%%%%%%%%%%%%%%%%%%%%%%%%%%%%%%%%%%%%%%%%%%%%%%%%%%%%%%
\section{Two-phase Algorithm}\label{sec:Algorithm3}
This section proposes a $k$-restricted two-phase heuristic ($k$-TPH) which is described in \textbf{Algorithm~\ref{Alg1}}. Let $T^t$ be the terminal-spanning tree at the end of iteration $t$ and let $K_t$ be the chosen full component at the end of iteration $t$. The first phase finds a terminal-spanning tree $T_{base}$ such that no full component can improve it. The terminal-spanning tree $T_{base}$ is a based criterion for the second phase. 
We denote by $S_{1}$ the solution in the first phase, and by $S_{2}$ the solution in the second phase. The first phase is a loss-contracting algorithm. The criterion function of $K$ with respect to $T^{t-1}$ is defined as
\begin{eqnarray*}
r=\frac{gain_{T^{t-1}}(K)}{loss(K)}.
\end{eqnarray*}

A chosen full component $K_i$ may be modified by other chosen full component. Assume that some edges $\left\{e_1,e_2,\ldots\right\}$ in $T^{t-1}$ that are corresponding to $\mathcal{C}[K_i]$ are deleted when adding $\mathcal{C}[K_t]$ to $T^{t-1}$. Some components are obtained by $K_i-\left\{e_1,e_2,\ldots\right\}$ and each component can be replaced by a full component with same terminals. The full component $K_i$ is replaced by these full components. That is because we want to ensure that $\frac{1}{2}\cdot cost(S_{1})\leq cost(T_{base})$.
If no edge in $T^{t-1}$ is corresponding to $\mathcal{C}[K_i]$, we keep a \textit{basic component} from $K_i$ that is a Steiner point directly connect to two terminals in which an edge belongs to $Loss(K_i)$ and another edge belongs to $K_i-Loss(K_i)$ (see Figure~\ref{fig3}). It guarantees that the chosen full components never be chosen again.  
However, it may bring that some Steiner points are leaves in $S_{1}$. Fortunately, these Steiner points can be removed. Therefore, this paper assume that no Steiner point is leaf in $S_{1}$.

%Assume an edge $e=\left\{b,c\right\}$ is deleted where $b$ is a terminal and $c$ is a Steiner point. If the Steiner point $c$ directly connect to only two terminals $a$ and $b$, we keep this component $\left\{a,b,c\right\}$ (see Figure~\ref{fig3}(b)). Figure~\ref{fig3} shows a chosen full component has been modified, and two full components replace this full component. It guarantees that the chosen full components never be chosen again.  
%However, it may bring that some Steiner points is leaves in $S_{1}$. Fortunately, these Steiner points can be removed. Therefore, this paper assume that no Steiner point is leaf in $S_{1}$.

\begin{figure}[t]
    \centering
        \includegraphics[width=3 in]{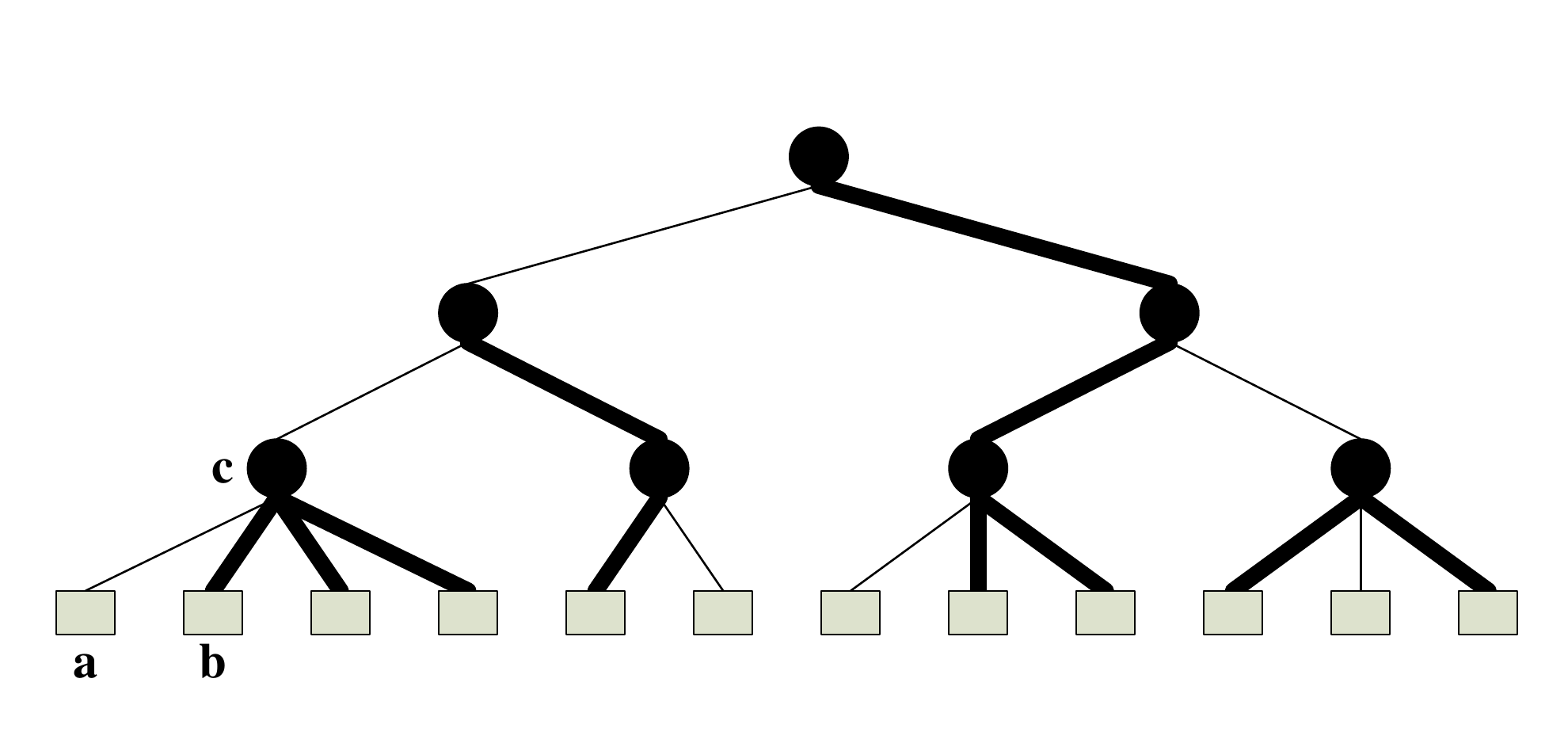} 
        \caption{A full component $K$: squares denote terminals, circles denote Steiner and bold black edges indicate $K-Loss(K)$. A subgraph $B=\left(\left\{a,b,c\right\},\left\{\left\{a,c\right\},\left\{b,c\right\}\right\}\right)$ is a basic component in $K$ where an edge $\left\{a,c\right\}$ belongs to $Loss(K)$ and another edge $\left\{b,c\right\}$ belongs to $K-Loss(K)$.}
    \label{fig3}
\end{figure}

The second phase calls the $k$-restricted enhanced relative greedy heuristic ($k$-ERGH), which is described in \textbf{Algorithm~\ref{Alg2}}, to obtain a Steiner tree $S_{2}$.
The $k$-ERGH iteratively finds a full component $K$ for modifying the terminal-spanning trees $T_{origin}^{0}=MST(G_{R})$ and $T_{base}^{0}$. When a full component $K_{t}$ has been chosen, the algorithm contracts the cost of the corresponding edges in $T_{origin}^{t-1}$ to zero, that is, $T_{origin}^{t}=MST(T_{origin}^{t-1}\cup E_{0}(K_{t}))$. Similarly, $T_{base}^{t}=MST(T_{base}^{t-1}\cup E_{0}(K_{t}))$.
The criterion function of $K$ with respect to $T_{origin}^{t-1}$ and $T_{base}^{t-1}$ is defined as
\begin{eqnarray*}
f(K)=\frac{load_{T_{base}^{t-1}}(K)}{\Psi_{T_{origin}^{t-1},T_{base}^{t-1}}(K)}.
\end{eqnarray*}

The following steps analyze the complexity of $k$-TPH. Recall that, $n$ is the number of terminals.
In the first phase, the number of iterations cannot exceed the number of full Steiner components $O(n^k)$.
The gain of a full component $K$ can be found in time $O(k)$ after precomputing the longest edges between any pair of nodes in the current minimum spanning tree, which may be accomplished in time $O(n \log n)$~\cite{Robins05}.
Thus, the runtime of all the iterations in the first phase can be bounded by $O(k n^{2k+1}\log n)$.
We also can obtain the runtime of all the iterations in the second phase is bounded by $O(k n^{2k+1}\log n)$.
Thus, the total runtime is $O(k n^{2k+1}\log n)$.
\begin{algorithm}
\caption{The $k$-restricted two-phase heuristic ($k$-TPH)}
    \begin{algorithmic}[1]
      \STATE \textbf{--------------------The first phase--------------------}
      \STATE $T^{0}=MST(G_{S})$
      \FOR{$t=1,2,\ldots$}
          \STATE Find a $k$-restricted full component $K_{t}=K$ with maximizes 
	        \[
                     r=\frac{gain_{T^{t-1}}(K)}{loss(K)}
          \]
          \IF{$r\leq 0$}
              \STATE $T_{base} = T^{t-1}$ and exit for-loop
          \ENDIF
          \IF{there exist some edges $\left\{e_1,e_2,\ldots\right\}\subseteq T^{t-1}-MST(T^{t-1}\cup E_0(K_t))$ and $\left\{e_1,e_2,\ldots\right\}\subseteq \mathcal{C}(K_i)$ for $i\neq t$}
          	\STATE Some components are obtained by $K_i-\left\{e_1,e_2,\ldots\right\}$ and each components can be replaced by a full component with same terminals.
          	\STATE Replaced the full component $K_i$ by these full components.
          	\STATE (for convenient to describe algorithm, we reuse the notain $K_i$ to represent these full components.)
          \ENDIF        
          \STATE $T^{t}=MST(T^{0}\cup \mathcal{C}[K_1]\cup \cdots \mathcal{C}[K_t])$
          \STATE $S_{1}=MST(T^{0}\cup K_1\cup \cdots \cup K_{t})$
          \IF{no edge in $T^{t}$ is corresponding to $\mathcal{C}[K_i]$ for $i\neq t$}  
            \STATE Keep a basic component from $K_i$.
            \STATE (we also reuse the notain $K_i$ to represent this basic component.)
          \ENDIF
      \ENDFOR
      \STATE \textbf{--------------------The second phase--------------------}
      \STATE $S_{2} = k\mbox{-ERGH}(T_{base})$
      \RETURN the minimum-cost tree $S$ between $S_{1}$ and $S_{2}$.
   \end{algorithmic}
\label{Alg1}
\end{algorithm}

\begin{algorithm}
\caption{The $k$-restricted enhanced relative greedy heuristic ($k$-ERGH)}
    \begin{algorithmic}[1]
      \REQUIRE $T_{base}$.
      \STATE $T_{base}^{0}=T_{base}$ and $T_{origin}^{0}=MST(G_{S})$
      \FOR{$t=1,2,\ldots$}
    	  \STATE Find a $k$-restricted full component $K_{t}=K$ which minimizes 
	\[
f(K)=\frac{load_{T_{base}^{t-1}}(K)}{\Psi_{T_{origin}^{t-1},T_{base}^{t-1}}(K)}
   \]
    	  \STATE $T_{origin}^{t}=MST(T_{origin}^{t-1}\cup E_{0}(K_{t}))$ 
    	  \STATE $T_{base}^{t}=MST(T_{base}^{t-1}\cup E_{0}(K_{t}))$
    	  \IF{$c(T_{origin}^{t}) = c(T_{base}^{t})$}
    	      \RETURN $MST(T_{origin}^{0}\cup K_1 \cup K_2 \cdots \cup K_t)$
    	  \ENDIF
    	\ENDFOR
   \end{algorithmic}
\label{Alg2}
\end{algorithm}

\section{Approximation ratio of the $k$-TPH}
This section shows the approximation result of the $k$-TPH. When a full component $K$ has been chosen, the following lemma shows that the first phase never choose the full component $K$ even it has been replaced by some full components.
\begin{lem}\label{lem09}
The first phase never choose the chosen full components again.
\end{lem}
\begin{proof}
Assume that the first phase choose a full component $K_t=K$. 
If $MST(T\bigcup_{i=1}^{t'}\mathcal{C}[K_i])$ contain all edges $e\in \mathcal{C}[K]$ in the iteration $t'>t$, $gain_{T}(K)\leq 0$ and the first phase never choose the full component $K$ again. If $MST(T\bigcup_{i=1}^{t'}\mathcal{C}[K_i])$ does not contain some edge $e\in \mathcal{C}[K]$ in the iteration $t'>t$, the edge $e$ has been improved by a chosen full component. The full component $K$ is divided into two components by removing the edge $e$.
Let $A$ and $B$ be two connected components of $K-\left\{e\right\}$. The full component $K_t$ is replaced by two full components $K_A$ and $K_B$ with terminals sets $\Gamma(A)$ and $\Gamma(B)$, respectively. We have $T^{t'}=MST(T\bigcup_{i=1}^{t-1}\mathcal{C}[K_i]\cup K_A \cup K_B \bigcup_{i=t+1}^{t'}C[K^i])$, $gain_{T^{t'}}(K)\leq gain_{T^{t'}}(A\cup B)\leq gain_{T^{t'}}(A)+gain_{T^{t'}}(B)$ and $loss(K)=loss(A)+loss(B)$. Finally,
\begin{eqnarray*}
\frac{gain_{T^{t'}}(K)}{loss(K)} & \leq & \frac{gain_{T^{t'}}(A))+gain_{T^{t'}}(B)}{loss(B)+loss(B)} \\
                     & \leq & \max\left\{\frac{gain_{T^{t'}}(A)}{loss(A)},\frac{gain_{T^{t'}}(B)}{loss(B)}\right\}.
\end{eqnarray*}
We show that $K_A$ never be replaced by $A$.
We knows that $cost(K_A)\leq cost(A)$ and $gain_{T^{t'}}(K_A)\leq 0$. The full component $K_A$ is superior to $A$.
We also can obtain that $K_B$ never be replaced by $B$.
The first phase never choose the full component $K$ again.

If no edge in $T^{t'}$ is corresponding to $\mathcal{C}[K]$, we keep a basic component in $K$. 
Then, we can find a full component that superior to $K$. The chosen full components never be chosen again.
\end{proof}

\begin{lem}\label{lem01}
$cost(T_{base}^{0} ) \geq \frac{1}{2}\cdot cost(S_{1})$. 
\end{lem}
\begin{proof}
%Let $S_{1}$ consist of full components $K_{j}$.
The cost of the Steiner tree in the first phase is
\begin{eqnarray*}
cost(S_{1}) & = & cost(T_{base}^{0})+ \sum_{K_{j}\in S_{1}} loss(K_{j}).
\end{eqnarray*}
Since $loss(K) \leq \frac{1}{2}\cdot cost(K)$~\cite{Robins05} for any full component $K$,
\begin{eqnarray*}
cost(S_{1}) & \leq  & cost(T_{base}^{0})+ \sum_{K_{j}\in S_{1}} \frac{1}{2}\cdot cost(K_{j}) \\
            & \leq & cost(T_{base}^{0})+ \frac{1}{2}\cdot cost(S_{1})
\end{eqnarray*}
which yields $cost(T_{base}^{0}) \geq \frac{1}{2}\cdot cost(S_{1})$.
\end{proof}

\begin{lem}\label{lem06}
If no full component can improve the terminal-spanning tree $T$,
	\[
	load_{T}\left(\bigcup_{i=1}^{n} K_{i}\right) \geq \sum_{i=1}^{n} load_{T}(K_{i})
\]
for full components $K_1,K_2,\ldots,K_n$.
\end{lem}
\begin{proof}
The proof can be obtained by the following chain of inequalities:
\begin{eqnarray*}
load_{T}\left(\bigcup_{i=1}^{n} K_{i}\right) 
   & =    & cost\left(\bigcup_{i=1}^{n} K_{i}\right)+mst\left(T\cup \bigcup_{i=1}^{n} E_{0}\left( K_{i}\right)\right)-cost(T)\\
   & =    & \sum_{i=1}^{n} cost(K_{i})+mst\left(T\cup \bigcup_{j=1}^{i} E_{0}\left( K_{j}\right)\right)-cost\left(T/\bigcup_{j=1}^{i-1} E_{0}\left( K_{j}\right)\right)\\
   & \geq & \sum_{i=1}^{n} cost(K_{i})+mst\left(T\cup E_{0}\left(K_{i}\right)\right)-cost(T)  \\
   & =    & \sum_{i=1}^{n} load_{T}(K_{i}).
\end{eqnarray*}
\end{proof}

The following lemma guarantees that the solution of $k$-TPH at the second phase is a well solution. 
\begin{lem}\label{lem03}
For any chosen full components $K_{i}$ and $K_{j}$, $\left|\Gamma(K_{i})\cap \Gamma(K_{j})\right|\leq 1$.
\end{lem}
\begin{proof}
Without loss of generality, assume that $\left|\Gamma(K_{i})\cap \Gamma(K_{j})\right|= 2$ and $j< i$. Both $T_{origin}^{i-1}-MST(T_{origin}^{i-1}\cup E_{0}(K_i))$ and $T_{base}^{i-1}-MST(T_{base}^{i-1}\cup E_{0}(K_i))$ contain a zero-cost edge that is from $E_{0}(K_{j})$. Since any full component cannot improve $T_{base}^{0}$, $MST(T_{base}^{0}\cup K) = T_{base}^{0}\cup Loss(K)$ for any full component $K$.
We can find a edge $e\in K_{i}-Loss(K_{i})$ such that $\Psi_{T_{origin}^{i-1},T_{base}^{i-1}}(K_i)=\Psi_{T_{origin}^{i-1},T_{base}^{i-1}}(A)+\Psi_{T_{origin}^{i-1},T_{base}^{i-1}}(B)$ and $load_{T_{base}^{i-1}}(K_i)\geq load_{T_{base}^{i-1}}(A \cup B) \geq load_{T_{base}^{i-1}}(A)+load_{T_{base}^{i-1}}(B)$ (from Lemma~\ref{lem06}) where $A$ and $B$ are two connected components of $K_{i}-\left\{e\right\}$. Finally,
\begin{eqnarray*}
\frac{load_{T_{base}^{i-1}}(K_i)}{\Psi_{T_{origin}^{i-1},T_{base}^{i-1}}(K_i)}\geq \frac{load_{T_{base}^{i-1}}(A)+load_{T_{base}^{i-1}}(B)}{\Psi_{T_{origin}^{i-1},T_{base}^{i-1}}(A)+\Psi_{T_{origin}^{i-1},T_{base}^{i-1}}(B)} \geq \min\left\{\frac{load_{T_{base}^{i-1}}(A)}{\Psi_{T_{origin}^{i-1},T_{base}^{i-1}}(A)},\frac{load_{T_{base}^{i-1}}(B)}{\Psi_{T_{origin}^{i-1},T_{base}^{i-1}}(B)}\right\}
\end{eqnarray*}
which contradicts the choice of $K_{i}$.
\end{proof}

\begin{lem}\label{lem04}
For any Steiner tree $S$, $load_{T_{base}^{0}}(S) \geq load_{T_{base}^{i-1}}\left(S/\bigcup_{j=1}^{i-1}E_{0}\left(K_{i}\right)\right)$.
\end{lem}
\begin{proof}
Since no full component can improve the terminal-spanning tree $T_{base}^{0}$, $cost(S)-cost(T_{base}^{0})-mst\left(S\cup \bigcup_{j=1}^{i-1} E_{0}\left(K_{i}\right)\right)+mst\left(T_{base}^{0}\cup \bigcup_{j=1}^{i-1} E_{0}\left(K_{i}\right)\right)\geq 0$. The proof can be obtained by the following chain of inequalities:
\begin{eqnarray*}
load_{T_{base}^{0}}(S)
   & =    & cost(S)-cost(T_{base}^{0}) \\
   & =    & mst\left(S\cup \bigcup_{j=1}^{i-1} E_{0}\left(K_{i}\right)\right)-mst\left(T_{base}^{0}\cup \bigcup_{j=1}^{i-1} E_{0}\left(K_{i}\right)\right) \\
   &      & + cost(S)-cost(T_{base}^{0})-mst\left(S\cup \bigcup_{j=1}^{i-1} E_{0}\left(K_{i}\right)\right)+mst\left(T_{base}^{0}\cup \bigcup_{j=1}^{i-1} E_{0}\left(K_{i}\right)\right)\\
   & \geq & mst\left(S\cup \bigcup_{j=1}^{i-1} E_{0}\left(K_{i}\right)\right)-mst\left(T_{base}^{0}\cup \bigcup_{j=1}^{i-1} E_{0}\left(K_{i}\right)\right) \\
   & =    & load_{T_{base}^{i-1}}\left(S/\bigcup_{j=1}^{i-1}E_{0}\left(K_{i}\right)\right).
\end{eqnarray*}
\end{proof}

\begin{lem}\label{lem11}
If $load_{T_{base}^{i-1}/E_{0}\left(C\right)}(K)\leq \Psi_{T_{origin}^{i-1},T_{base}^{i-1}/ E_{0}\left(C\right)}(K)$ for any full components $C$ and $K$, 
	\[
	\frac{load_{T_{base}^{i-1}/ E_{0}\left(C\right)}(K)}{\Psi_{T_{origin}^{i-1},T_{base}^{i-1}/ E_{0}\left(C\right)}(K)} \geq \frac{load_{T_{base}^{i-1}}(K)}{\Psi_{T_{origin}^{i-1},T_{base}^{i-1}}(K)}.
\]
\end{lem}
\begin{proof}
Since $load_{T_{base}^{i-1}/E_{0}\left(C\right)}(K)\leq \Psi_{T_{origin}^{i-1},T_{base}^{i-1}/ E_{0}\left(C\right)}(K)$ and $cost(T_{base}^{i-1}/ E_{0}\left(C\right))-mst(T_{base}^{i-1}\cup E_{0}\left(C\right)\cup E_{0}\left(K\right))\leq cost(T_{base}^{i-1})-mst(T_{base}^{i-1}\cup E_{0}\left(K\right))$,
the proof can be obtained by the following chain of inequalities:
\begin{eqnarray*}
   &      & \frac{load_{T_{base}^{i-1}/ E_{0}\left(C\right)}(K)}{\Psi_{T_{origin}^{i-1},T_{base}^{i-1}/ E_{0}\left(C\right)}(K)} \\
   & =    & \frac{cost(K)+mst(T_{base}^{i-1}\cup E_{0}\left(C\right)\cup E_{0}\left(K\right))-cost(T_{base}^{i-1}/ E_{0}\left(C\right))}{cost(T_{origin}^{i-1})-cost(T_{base}^{i-1}/ E_{0}\left(C\right))-mst(T_{origin}^{i-1}\cup E_{0}(K))+mst(T_{base}^{i-1}\cup E_{0}\left(C\right)\cup E_{0}(K))} \\
   & \geq    & \frac{cost(K)+mst(T_{base}^{i-1}\cup E_{0}\left(K\right))-cost(T_{base}^{i-1})}{cost(T_{origin}^{i-1})-cost(T_{base}^{i-1})-mst(T_{origin}^{i-1}\cup E_{0}(K))+mst(T_{base}^{i-1}\cup E_{0}(K))} \\
   & =    & \frac{load_{T_{base}^{i-1}}(K)}{\Psi_{T_{origin}^{i-1},T_{base}^{i-1}}(K)}.
\end{eqnarray*}
\end{proof}

Based on the analysis in~\cite{Zel96}, the bound on the cost of our solution is as follows.
\begin{thm}\label{thm:0001}
The $k$-TPH finds a Steiner tree $S$ such that 
	\[
cost(S)  \leq  \left(\ln \frac{mst -cost(T_{base}^{0})}{opt_{k} -cost(T_{base}^{0})}+ 1\right)\cdot \left(opt_{k}-cost(T_{base}^{0})\right) + cost(T_{base}^{0}).
\]
\end{thm}
\begin{proof}
Let $M_i=cost(T_{origin}^i)-cost(T_{base}^i)$ and $m_i=M_{i-1}-M_{i}$. Therefore, $f(K_i)=\frac{load_{T_{base}^{i-1}}(K_i)}{m_i}$. 
Let $Opt_{k}^{i-1}=\left(Opt_{k}/\bigcup_{l=1}^{i-1} E_{0}\left(k_{l}\right)\right)-\bigcup_{l=1}^{i-1} E_{0}\left(K_{l}\right)$. For $i=1,\ldots, r+1$ and $load_{T_{base}^{0}}(Opt_{k})\leq M_{i-1}$, we have
\begin{eqnarray*}
\frac{load_{T_{base}^{0}}(Opt_{k})}{M_{i-1}} 
   & = & \frac{load_{T_{base}^{0}}(Opt_{k})}{\Psi_{T_{origin}^{i-1},T_{base}^{i-1}}(Opt_{k})} \\
   & \overset{\mbox{Lem~\ref{lem04}}}{\geq} & \frac{load_{T_{base}^{i-1}}(Opt_{k}^{i-1})}{\Psi_{T_{origin}^{i-1},T_{base}^{i-1}}(Opt_{k}^{i-1})} \\
   & = & \frac{\sum_{X_{j}\in Opt_{k}^{i-1}} load_{T_{base}^{i-1}/ \bigcup_{l=1}^{j-1} E_{0}\left(X_{l}\right)}(X_{j})}{\sum_{X_{j}\in Opt_{k}^{i-1}} \Psi_{T_{origin}^{i-1}/ \bigcup_{l=1}^{j-1} E_{0}\left(X_{l}\right),T_{base}^{i-1}/ \bigcup_{l=1}^{j-1} E_{0}\left(X_{l}\right)}(X_{j})} \\
   & \geq & \frac{\sum_{X_j\in Opt_{k}^{i-1}} load_{T_{base}^{i-1}/ \bigcup_{l=1}^{j-1} E_{0}\left(X_{l}\right)}(X_j)}{\sum_{X_j\in Opt_{k}^{i-1}} \Psi_{T_{origin}^{i-1},T_{base}^{i-1}/ \bigcup_{l=1}^{j-1} E_{0}\left(X_{l}\right)}(X_j)} \\
   & \overset{\mbox{Lem~\ref{lem11}}}{\geq} & \frac{\sum_{X_j\in Opt_{k}^{i-1}} load_{T_{base}^{i-1}}(X_j)}{\sum_{X_j\in Opt_{k}^{i-1}} \Psi_{T_{origin}^{i-1},T_{base}^{i-1}}(X_j)} \\
   %& \geq & \frac{\sum_{X_j\in Opt_{k}^{i-1}} load_{T_{base}^{i-1}}(X_j)}{\sum_{X_j\in Opt_{k}^{i-1}} \Psi_{T_{origin}^{i-1},T_{base}^{i-1}}(X_j)} \\   
   & \geq & \min_{X_j\in Opt_{k}^{i-1}} \left\{\frac{load_{T_{base}^{i-1}}(X_j)}{\Psi_{T_{origin}^{i-1},T_{base}^{i-1}}(X_j)}\right\} \\
   & \geq & \frac{load_{T_{base}^{i-1}}(K_{i})}{m_i}.
\end{eqnarray*}
Replacing $m_i=M_{i-1}-M_{i}$ into the above inequality yields
\begin{eqnarray}\label{eq:1}
M_{i} \leq M_{i-1}\left(1-\frac{load_{T_{base}^{i-1}}(K_{i})}{load_{T_{base}^{0}}(Opt_{k})}\right)
\end{eqnarray}
for $i = 1, 2, \ldots, t$. From the inequality~(\ref{eq:1}), 
\begin{eqnarray*}
M_{r} & \leq & M_{0}\prod_{i=1}^{t}\left(1-\frac{load_{T_{base}^{i-1}}(K_{i})}{load_{T_{base}^{0}}(Opt_{k})}\right).
\end{eqnarray*}
Taking the natural logarithms of both sides and using the inequality $\ln(1+x)\leq x$,
\begin{eqnarray}\label{eq:001}
\ln \frac{M_{0}}{M_{r}} & \geq & -\sum_{i=1}^{t} \ln \left(1-\frac{load_{T_{base}^{i-1}}(K_{i})}{load_{T_{base}^{0}}(Opt_{k})}\right) \notag \\
                        & \geq & \frac{\sum_{i=1}^{t} load_{T_{base}^{i-1}}(K_{i})}{load_{T_{base}^{0}}(Opt_{k})}.
\end{eqnarray}
Since $k$-TPA interrupts at $M_{t}=c(T_{origin}^{t})-c(T_{base}^{t})=0$, there exists $M_{r}>load_{T_{base}^{0}}(Opt_{k}) \geq M_{r+1}$ for some $r< t$.

The value $m_{r+1}$ can be split into two values $m^*$ and $m'$ such that 
\begin{eqnarray}\label{enq:0003}
m^*=M_{r}-load_{T_{base}^{0}}(Opt_{k}),
\end{eqnarray}
\begin{eqnarray}\label{enq:0004}
m'=load_{T_{base}^{0}}(Opt_{k})-M_{r+1},
\end{eqnarray}
According to inequality (\ref{enq:0003}), we have
\begin{eqnarray}\label{enq:0005}
M_{r+1}^*=M_{r}-m^*=M_{r}-M_{r}+load_{T_{base}^{0}}(Opt_{k})=load_{T_{base}^{0}}(Opt_{k}).
\end{eqnarray}

The value $load_{T_{base}^{r}}(K_{r+1})$ also can be split into $w^*$ and $w'$ such that $\frac{load_{T_{base}^{r}}(K_{r+1})}{m_{r+1}} = \frac{w^*}{m^*} = \frac{w'}{m'}$. Since $\frac{load_{T_{base}^{r}}(K_{r+1})}{m_{r+1}} = \frac{w^*}{m^*}$, inequality (\ref{eq:001}) implies that
\begin{eqnarray}\label{eq:002}
\ln \frac{M_{0}}{M_{r+1}^*} \geq \frac{\sum_{i=1}^{r} load_{T_{base}^{i-1}}(K_{i})+w^*}{load_{T_{base}^{0}}(Opt_{k})}.
\end{eqnarray}
Since $\frac{load_{T_{base}^{r}}(K_{r+1})}{m_{r+1}}\leq \frac{load_{T_{base}^{0}}(Opt_{k})}{M_{r}}\leq 1$, we have
\begin{eqnarray}\label{enq:0006}
w'\leq m'.
\end{eqnarray}

The ratio related to the cost of approximate Steiner tree after $r + 1$ iterations is at most
\begin{eqnarray*}
\frac{cost(S_{2})-cost(T_{base}^{0})}{opt_{k}-cost(T_{base}^{0})} 
   & =    & \frac{mst(T_{origin}^{0}\cup\bigcup_{i=1}^{t} K_i)-cost(T_{base}^{0})}{load_{T_{base}^{0}}(Opt_{k})} \\
   & \overset{\mbox{Lem~\ref{lem03}}}{\leq} & \frac{\sum_{i=1}^{r+1} load_{T_{base}^{i-1}}(K_{i}) +M_{r+1}}{load_{T_{base}^{0}}(Opt_{k})} \\
   & =    & \frac{\sum_{i=1}^{r} load_{T_{base}^{i-1}}(K_{i}) +w^*+w' +M_{r+1}}{load_{T_{base}^{0}}(Opt_{k})} \\
   & \overset{\mbox{(\ref{eq:002})}}{\leq} & \ln \frac{M_{0}}{M_{r+1}^*}+ \frac{w' +M_{r+1}}{load_{T_{base}^{0}}(Opt_{k})} \\
   & \overset{\mbox{(\ref{enq:0006})}}{\leq} & \ln \frac{M_{0}}{M_{r+1}^*}+ \frac{m' +M_{r+1}}{load_{T_{base}^{0}}(Opt_{k})} \\
   & \overset{\mbox{(\ref{enq:0004})}}{=}   & \ln \frac{M_{0}}{M_{r+1}^*}+ 1 \\
   & \overset{\mbox{(\ref{enq:0005})}}{=}   & \ln \frac{M_{0}}{load_{T_{base}^{0}}(Opt_{k})}+ 1 \\   
   & =    & \ln \frac{cost(T_{origin}^{0}) -cost(T_{base}^{0})}{opt_{k} -cost(T_{base}^{0})}+ 1 \\
  % & \overset{\mbox{Lem~\ref{lem08}}}{\leq} & \ln \left(\frac{cost(T_{origin}^{0}) -opt_{k}+loss_{k}}{opt_{k}-opt_{k}+loss_{k}}\right)+ 1 \\
   & =    & \ln \frac{mst -cost(T_{base}^{0})}{opt_{k} -cost(T_{base}^{0})}+ 1
\end{eqnarray*}
which yields
\begin{eqnarray}\label{enq:0007}
cost(S) \leq cost(S_{2})   \leq   \left(\ln \frac{mst -cost(T_{base}^{0})}{opt_{k} -cost(T_{base}^{0})}+ 1\right)\cdot \left(opt_{k}-cost(T_{base}^{0})\right) + cost(T_{base}^{0}).
\end{eqnarray}
\end{proof}

Since $cost(T_{base}^{0})\leq opt_k$ (from Lemma~\ref{lem08}) and $cost(T_{base}^{0})\geq \frac{1}{2}\cdot cost(S_1) \geq \frac{1}{2}\cdot opt_k$ (from Lemma~\ref{lem01}), we can assume that $cost(T_{base}^{0}) = \alpha \cdot opt_k$ for $\alpha \in  \left(\frac{1}{2},1\right)$.
The following result can be obtained.

\begin{thm}\label{thm:0003}
If $cost(T_{base}^{0}) = \alpha \cdot opt_k$ for $\alpha \in \left(\frac{1}{2},1\right)$, the $k$-TPH finds a Steiner tree $S$ such that 
\begin{eqnarray*}
cost(S)  \leq   \left(\ln \frac{mst -\alpha \cdot opt_k}{opt_{k} -\alpha \cdot opt_k}+ 1\right)\cdot \left(opt_{k}-\alpha \cdot opt_k\right) + \alpha \cdot opt_k.
\end{eqnarray*}
and
\begin{eqnarray*}
cost(S) \leq   2\cdot \alpha \cdot opt_k.
\end{eqnarray*}
\end{thm}
\begin{proof}
From Theorem \ref{thm:0001}, we have
\begin{eqnarray*}
cost(S)  \leq   \left(\ln \frac{mst -\alpha \cdot opt_k}{opt_{k} -\alpha \cdot opt_k}+ 1\right)\cdot \left(opt_{k}-\alpha \cdot opt_k\right) + \alpha \cdot opt_k.
\end{eqnarray*}
According to Lemma~\ref{lem01}, $cost(S)  \leq 2\cdot cost(T_{base}^{0}) =2\cdot \alpha \cdot opt_k$.
\end{proof}

\section{Performance of the $k$-TPH in general graphs}
The following corollaries gives a bound on the cost of the Steiner tree generated by $k$-TPH. 
\begin{coro}\label{thm0006}
The $k$-TPH has an approximation ratio of at most $1.4295$.
\end{coro}
\begin{proof}
We have $mst \leq 2 \cdot opt$ (see~\cite{Takahashi80}). Theorem~\ref{thm:0003} yield
\begin{eqnarray*}
\frac{cost(S)}{opt} & \leq &  \left(\ln \frac{2 \cdot opt -\alpha \cdot opt_k}{opt_{k} -\alpha \cdot opt_k}+ 1\right)\cdot \left(1-\alpha\right)\frac{opt_{k}}{opt} + \alpha \cdot \frac{opt_{k}}{opt} \\
                    &  =   & \left(\ln \frac{\frac{2}{\rho_k} -\alpha}{1 -\alpha}+ 1\right)\cdot \left(1-\alpha\right)\rho_k + \alpha \cdot \rho_k
\end{eqnarray*}
and
\begin{eqnarray*}
\frac{cost(S)}{opt} \leq   2\cdot \alpha \cdot \rho_k,
\end{eqnarray*}
where $\rho_k$ is the worst-case ratio of $\frac{opt_{k}}{opt}$. Borchers and Du~\cite{Borchers97} show that $\rho_k \leq 1+\left\lfloor \log_{2}k\right\rfloor^{-1}$ and $\lim_{k\rightarrow \infty}\rho_k =1$. When $k\rightarrow \infty$, the approximation ratio of the $k$-TPH converges to
\begin{eqnarray*}
A(\alpha) = \left(\ln \frac{2 -\alpha}{1 -\alpha}+ 1\right)\cdot \left(1-\alpha\right) + \alpha.
\end{eqnarray*}
and
\begin{eqnarray*}
B(\alpha) =   2\cdot \alpha.
\end{eqnarray*}
Since $A(\alpha)$ is decreasing in $\alpha$ and $B(\alpha)$ is increasing in $\alpha$, solving $A(\alpha)=B(\alpha)$ yeilds $\alpha^*\approx 0.7147$. The $k$-TPH has an approximation ratio of at most $A(\alpha^*)\approx 1.4295$.
\end{proof}

%\bibliography{papers}
%\bibliographystyle{IEEEannot}

\end{document}